\documentclass[a4paper]{llncs}
\usepackage{IEEEtrantools}
\usepackage[utf8]{inputenc}
\usepackage{amssymb}
\usepackage{amsmath}
\usepackage{bm}
\usepackage{mathrsfs}
\usepackage{xcolor}
\usepackage{scalerel}
\usepackage{theorem}
\usepackage{cite}
\usepackage{url}
\usepackage{array}
\usepackage{multirow}
\usepackage{pgfplots}
\usepackage{hf-tikz}
\usepackage[ruled,vlined,titlenumbered,linesnumbered]{algorithm2e}
\usepackage[colorlinks=true,citecolor=blue,linkcolor=blue,urlcolor=blue]{hyperref}
\usepackage[nolist]{acronym}
\usepackage{setspace}
\usepackage[normalem]{ulem}
\usepackage{orcidlink}

\usepackage{mathtools}
\mathtoolsset{showonlyrefs}

\usetikzlibrary{plotmarks,patterns,tikzmark}
\usetikzlibrary{matrix,decorations.pathreplacing,calc}

\pgfplotsset{
compat=1.17,
mystyle/.style={
    scale only axis,
    width=0.7\columnwidth,
    height=0.5\columnwidth,
    label style={inner sep=0, font=\normalsize},
    tick label style={font=\scriptsize},
    legend style={font=\scriptsize},
    mark size=3,
    major grid style={dashed},
    line width=0.8pt,
    axis line style = thin}
}

\SetKwComment{Comment}{}{}

\newcommand{\defeq}{:=}

\newcommand{\defemph}[1]{\emph{#1}}

\newcommand{\set}[1]{\ensuremath{\mathcal{#1}}}

\newcommand{\NN}{\ensuremath{\mathbb{N}}}

\newcommand{\OCompl}[1]{\ensuremath{\mathcal{O}({#1})}}

\newcommand{\Fq}{\ensuremath{\mathbb F_{q}}}
\newcommand{\Fqm}{\ensuremath{\mathbb F_{q^m}}}

\newcommand{\Polyring}{\ensuremath{\Fqm[x]}}

\newcommand{\aut}{\ensuremath{\sigma}}
\newcommand{\autinv}{\ensuremath{\aut^{-1}}}
\newcommand{\der}{\ensuremath{\delta}}

\newcommand{\conj}[2]{\ensuremath{{#1}^{#2}}}

\newcommand{\SkewPolyring}{\ensuremath{\Fqm[x;\aut,\der]}}
\newcommand{\SkewPolyringZeroDer}{\ensuremath{\Fqm[x;\aut]}}

\newcommand{\opev}[3]{\ensuremath{{#1}(#2)_{#3}}}

\newcommand{\op}[2]{\ensuremath{\mathcal{D}_{#1}(#2)}}
\newcommand{\opLargeParens}[2]{\ensuremath{\mathcal{D}_{#1}\Bigl(#2\Bigr)}}
\newcommand{\opexp}[3]{\ensuremath{\mathcal{D}_{#1}^{#3}(#2)}}
\newcommand{\opexpinv}[3]{\ensuremath{\left(\mathcal{D}_{#1}^{#3}\right)^{-1}(#2)}}

\newcommand{\opfullparam}[3]{\ensuremath{\mathcal{D}_{#1}^{#3}(#2)}}
\newcommand{\opfullparamexp}[4]{\ensuremath{\left(\mathcal{D}_{#1}^{#3}\right)^{#4}(#2)}}

\newcommand{\genNorm}[2]{\ensuremath{\mathcal{N}_{#1}\left(#2\right)}}

\DeclareMathOperator{\wt}{wt}
\DeclareMathOperator{\rk}{rk}

\DeclareMathOperator{\diag}{diag}

\DeclareMathOperator{\GL}{GL}
\DeclareMathOperator{\Aut}{Aut}
\DeclareMathOperator{\Gal}{Gal}
\DeclareMathOperator{\Sym}{Sym}

\renewcommand{\vec}[1]{
	{\mathchoice{\mbox{\boldmath$\displaystyle #1$}}
	{\mbox{\boldmath$\textstyle #1$}}
	{\mbox{\boldmath$\scriptstyle #1$}}
	{\mbox{\boldmath$\scriptscriptstyle #1$}}}
}
\newcommand{\mat}[1]{\ensuremath{\bm{#1}}}

\renewcommand{\a}{\vec{a}}
\renewcommand{\b}{\vec{b}}
\renewcommand{\c}{\vec{c}}

\newcommand{\g}{\vec{g}}

\newcommand{\n}{\vec{n}}

\renewcommand{\v}{\vec{v}}

\newcommand{\x}{\vec{x}}
\newcommand{\y}{\vec{y}}

\newcommand{\vecalpha}{\ensuremath{\boldsymbol{\alpha}}}
\newcommand{\vecbeta}{\ensuremath{\boldsymbol{\beta}}}

\newcommand{\A}{\mat{A}}
\newcommand{\B}{\mat{B}}

\newcommand{\G}{\mat{G}}
\renewcommand{\H}{\mat{H}}

\newcommand{\M}{\mat{M}}

\renewcommand{\S}{\mat{S}}

\newcommand{\0}{\ensuremath{\mathbf 0}}
\DeclareMathOperator{\Id}{Id}

\newcommand{\opVandermonde}[3]{\ensuremath{\mat{V}_{#1}(#2)_{#3}}}

\newcommand{\opMoore}[3]{\ensuremath{\mathfrak{M}_{#1}(#2)_{#3}}}

\newcommand{\SumRankWeight}{\ensuremath{\wt_{\Sigma R}}}

\newcommand{\SumRankDist}{d_{\ensuremath{\Sigma}R}}

\newcommand{\mycode}[1]{\ensuremath{\mathcal{#1}}}

\newcommand{\linRS}[1]{\ensuremath{\mathrm{LRS}[#1]}}

\newcommand{\GLRS}[1]{\ensuremath{\mathrm{GLRS}[#1]}}

\newcommand{\len}{\ensuremath{n}}
\newcommand{\lenShot}[1]{\ensuremath{\len_{#1}}}
\newcommand{\shot}[2]{\ensuremath{{#1}^{(#2)}}}
\newcommand{\shots}{\ensuremath{\ell}}

\newcommand{\linIsometries}{\mathrm{LI}(\Fqm^{n})}
\newcommand{\semilinIsometries}{\mathrm{SI}(\Fqm^{n})}
\newcommand{\actLin}{\mathrm{act}_{\mathrm{LI}}}
\newcommand{\actSemilin}{\mathrm{act}_{\mathrm{SI}}}

\begin{acronym}
	\acro{BCH}{Bose--Chaudhuri--Hocquenghem}
	\acro{BMD}{bounded minimum distance}
	\acro{ESP}{error span polynomial}
	\acro{ELP}{error locator polynomial}
	\acro{SRS}{skew Reed--Solomon}
	\acro{ISRS}{interleaved skew Reed--Solomon}
	\acro{LRS}{linearized Reed--Solomon}
	\acro{LLRS}{lifted linearized Reed--Solomon}
	\acro{ILRS}{interleaved linearized Reed--Solomon}
	\acro{LILRS}{lifted interleaved linearized Reed--Solomon}
	\acro{MDS}{maximum distance separable}
	\acro{MSRD}{maximum sum-rank distance}
	\acro{MSD}{maximum skew distance}
	\acro{RS}{Reed--Solomon}
	\acro{lclm}{least common left multiple}
	\acro{LEEA}{linearized extended Euclidean algorithm}
	\acro{SEEA}{skew extended Euclidean algorithm}
	\acro{GSE}{generalized skew-evaluation}
	\acro{GRS}{generalized Reed--Solomon}
	\acro{GSRS}{generalized skew Reed--Solomon}
	\acro{GLRS}{generalized linearized Reed--Solomon}
	\acro{SS}{Sidelnikov--Shestakov}
	\acro{SaS}{Sidelnikov and Shestakov}
	\acro{HMR}{Horlemann--Marshall--Rosenthal}
	\acro{HaMaR}{Horlemann-Trautmann, Marshall, and Rosenthal}
	\acro{KEM}{key-encapsulation mechanism}
	\acro{GPT}{Gabidulin--Paramonov--Tretjakov}
\end{acronym}

\begin{document}

    \title{
        Distinguishing and Recovering Generalized Linearized Reed--Solomon Codes
    }

    \author{Felicitas Hörmann$^{1,2}$\,\orcidlink{0000-0003-2217-9753} \and Hannes Bartz$^1$\,\orcidlink{0000-0001-7767-1513} \and Anna-Lena Horlemann$^2$\,\orcidlink{0000-0003-2685-2343}}

    \institute{
        $^1$~Institute of Communications and Navigation, German Aerospace Center (DLR), Germany \\
        \email{$\{$felicitas.hoermann, hannes.bartz$\}$@dlr.de}
        \\[.25cm]
        $^2$~School of Computer Science, University of St.\ Gallen, Switzerland \\
        \email{anna-lena.horlemann@unisg.ch}
    }

    \maketitle

	\begin{abstract}
		We study the distinguishability of \ac{LRS} codes by defining and analyzing analogs of the square-code and the Overbeck distinguisher for classical \acl{RS} and Gabidulin codes, respectively. Our main results show that the square-code distinguisher works for \ac{GLRS} codes defined with the trivial automorphism, whereas the Overbeck-type distinguisher can handle \ac{LRS} codes in the general setting. We further show how to recover defining code parameters from any generator matrix of such codes in the zero-derivation case. For other choices of automorphisms and derivations simulations indicate that these distinguishers and recovery algorithms do not work. The corresponding \ac{LRS} and \ac{GLRS} codes might hence be of interest for code-based cryptography.
	\end{abstract}

	\acresetall

	\section{Introduction}

	Researchers have made tremendous progress in the design and realization of quantum computers in the last decades.
	As it was shown that quantum computers are capable of solving both the prime-factorization and the discrete-logarithm
	problem in polynomial time, attackers can break most of today's public-key cryptosystems (as e.g.\ RSA and ECC) if
	they have a powerful quantum computer at hand.
	The urgent need for quantum-safe cryptography is obvious, especially since \emph{store now, harvest later} attacks
	allow to save encrypted data now and decrypt it as soon as the resources are available.
	This is reflected in the standardization process that NIST started for post-quantum cryptography in 2016.
	The first \acp{KEM} were standardized in July 2022 after three rounds of the competition and some of the submissions
	were forwarded to a fourth round for further investigation~\cite{NISTreport2022}.
	Three out of the four remaining \ac{KEM} candidates in round four are code-based.
	Moreover, the fourth one (namely, SIKE) was recently broken~\cite{wouter2022}.
	This explains why the community has high hopes and trust in coding-related primitives
	even though no code-based candidate has been chosen for standardization so far.

	Code-based cryptography mostly relies on McEliece-like schemes that are inspired by the seminal paper~\cite{McE1978}.
	The main idea is to choose a generator matrix of a secret code and disguise its algebraic structure by applying some,
	in most cases isometric or near-isometric, transformations such that an adversary cannot derive the known (or any other) efficient decoder from the
	mere knowledge of the scrambled matrix.

    McEliece-like instances based on a variety of code families and disguising functions in the Hamming and the rank
	metric were proposed over time.
	For example, the works~\cite{sidelnikov1992insecurity,BerLoi2005} are based on \ac{RS} codes in the Hamming metric and the \acs{GPT}
	system and its variants (see e.g.\ \cite{GabParTre1991,Gabidulin2008,Rashwan2010smartApproach}) use Gabidulin codes in the rank metric.
	But in both cases, polynomial-time attacks were proposed and broke several of the systems:
	\ac{RS} codes can be distinguished from random codes by using the square-code approach introduced
	in~\cite{sidelnikov1992insecurity,wieschebrink2010cryptanalysis} and Overbeck-like strategies~\cite{Ove2005,Ove2007,Ove2007PhD,horlemann2016rankBased,horlemann2018extension}
	yield a distinguisher for Gabidulin codes.
	The works also explain the recovery of an equivalent secret key which enables the attacker to decrypt with respect
	to the public code.

	The sum-rank metric was first established in 2005~\cite[Sec.~III]{lu2005unified} and generalizes both the Hamming and the
	rank metric.
	It is thus natural to investigate if McEliece-like cryptosystems based on sum-rank-metric codes can ensure secure
	communication.
	The work~\cite{PucRenRos2020} considers generic decoding of sum-rank-metric codes and hence gives guidance for the
	security-level estimation of sum-rank-based cryptography.
    Mart{\'\i}nez-Pe{\~n}as~\cite{Mar2018} introduced \ac{LRS} codes which are the sum-rank analogs of \ac{RS} and
	Gabidulin codes and thus could be a first naive choice for secret codes in McEliece-like systems.

	We focus on the task of distinguishing \ac{LRS} codes from random codes and present two distinguishers that are
	inspired by the square-code idea and by Overbeck's approach, respectively.
	Our results can be applied to distinguish \ac{GLRS} codes which we define as \ac{LRS} codes with nonzero
	block multipliers.
	As this more general code family is closed under semilinear equivalence, the methods also apply to \ac{GLRS} codes with isometric disguising.
	We finally focus on the zero-derivation case and show how an efficient decoding algorithm can be recovered from a \ac{GLRS}
	generator matrix that was disguised by means of semilinear isometries.

	\section{Preliminaries}
	
	Let us first gather some notions and results that we will use later on.
	In particular, let $q$ be a prime power and denote the finite field of order $q$ by $\Fq$.
	For $m \geq 1$, we further consider the extension field $\Fqm \supseteq \Fq$ of order $q^m$.
	For a matrix $\M  \in \Fqm^{k\times n}$, let $\langle M \rangle$ denote the $\Fqm$-linear vector space spanned by the rows of $\M$.

	\subsection{The Sum-Rank Metric}\label{subsec:sum-rank_metric}

	An \defemph{(integer) composition} of $\len \in \NN^{\ast}$ into $\shots \in \NN^{\ast}$ parts (or
	\emph{$\shots$-composition} for short) is a vector $\n = (\lenShot{1}, \dots, \lenShot{\shots}) \in \NN^{\shots}$ with
	$\lenShot{i} > 0$ for all $1 \leq i \leq \shots$ that satisfies $\len = \sum_{i=1}^{\shots} n_i$.
	If $\n$ contains $k$ distinct elements $\tilde{n}_1, \dots, \tilde{n}_k$, let $\lambda_j$ denote the number of
	occurrences of $\tilde{n}_j$ in $\n$ for each $j = 1, \dots, k$ and write $\lambda(\n) \defeq
	(\lambda_1, \dots, \lambda_k) \in \NN^{k}$.

	Throughout the paper, $n \in \NN^{\ast}$ usually refers to the length of the considered codes and we will stick to one
	particular $\shots$-composition $\n = (n_1, \dots, n_{\shots})$ of $n$.
	We often divide vectors $\x \in \Fqm^{n}$ or matrices $\M \in \Fqm^{k \times n}$ with $k \in \NN^{\ast}$ into blocks
	with respect to $\n$.
	Namely, we write $\x = (\shot{\x}{1} \mid \dots \mid \shot{\x}{\shots})$ with $\shot{\x}{i} \in \Fqm^{n_i}$ for all
	$1 \leq i \leq \shots$ and $\M = (\shot{\M}{1} \mid \dots \mid \shot{\M}{\shots})$ with $\shot{\M}{i} \in
	\Fqm^{k \times n_{i}}$ for all $1 \leq i \leq \shots$, respectively.

	The \defemph{sum-rank weight} of a vector $\x \in \Fqm^{n}$ (with respect to $\n$) is defined as
	$\SumRankWeight^{\n}(\x) = \sum_{i=1}^{\shots} \rk_q(\shot{\x}{i})$, where $\rk_q(\shot{\x}{i})$ is the maximum
	number of $\Fq$-linearly independent entries of $\shot{\x}{i}$ for $i = 1, \dots, \shots$.
	The hereby induced \defemph{sum-rank metric} is given by $\SumRankDist^{\n}(\x, \y) = \SumRankWeight^{\n}(\x - \y)$
	for $\x, \y \in \Fqm^{n}$.
	Since we always consider the same $\shots$-composition $\n$, we write $\SumRankWeight$ and $\SumRankDist$
	for simplicity.

	An \emph{$\Fqm$-linear sum-rank-metric code} $\mycode{C}$ is an $\Fqm$-subspace of $\Fqm^{n}$.
	Its length is $n$ and its dimension $k \defeq \dim(\mycode{C})$.
	We further define its \emph{minimum (sum-rank) distance} as
	\begin{equation}
	    d(\mycode{C}) \defeq \{ \SumRankDist(\c_1, \c_2) : \c_1, \c_2 \in \mycode{C}, \c_1 \neq \c_2 \} \\
		= \{ \SumRankWeight(\c) : \c \in \mycode{C}, \c \neq \0 \},
	\end{equation}
	where the last equality follows from the linearity of the code.
	A matrix $\G \in \Fqm^{k \times n}$ is a \emph{generator matrix} of $\mycode{C}$ if $\mycode{C} = \langle \G \rangle$.
	If $\mycode{C}$ is the kernel of a matrix $\H \in \Fqm^{(n-k) \times n}$, $\H$ is called a \emph{parity-check matrix}
	of $\mycode{C}$.
	The code generated by any parity-check matrix $\H$ of $\mycode{C}$ is the \emph{dual code} of $\mycode{C}$ and
	denoted by $\mycode{C}^{\perp}$.

	\subsection{Automorphisms, Derivations, and Conjugacy}

	An \defemph{automorphism} $\theta$ on $\Fqm$ is a mapping $\theta: \Fqm \to \Fqm$ with the properties
	$\theta(a + b) = \theta(a) + \theta(b)$ and $\theta(a \cdot b) = \theta(a) \cdot \theta(b)$ for all $a, b \in \Fqm$.
	We denote the group of all $\Fqm$-automorphisms by $\Aut(\Fqm)$.
	Note that every automorphism fixes a subfield of $\Fqm$ pointwise.
	The cyclic subgroup of $\Aut(\Fqm)$, whose elements fix (at least) $\Fq$ pointwise, is called the \emph{Galois group} of the
	field extension $\Fqm/\Fq$ and we denote it by $\Gal(\Fqm/\Fq)$.
	Every $\aut \in \Gal(\Fqm/\Fq)$ is a power of the \emph{Frobenius automorphism} (with respect to $q$) that is defined as
	\begin{equation}
		\varphi: \Fqm \to \Fqm,
		\quad a \mapsto a^q.
	\end{equation}
	Namely, $\aut \in \{\varphi^0, \dots, \varphi^{m - 1}\}$.
	The fixed field of $\aut = \varphi^{l}$ with $l \in \{0, \dots, m-1\}$ is $\Fq^{\gcd(l, m)}$.
	For simplicity, assume in the following that $\aut = \varphi^{l}$ with $\gcd(l, m) = 1$, i.e., let the fixed field of $\aut$ be $\Fq$.

	A \defemph{$\aut$-derivation} is a map $\der: \Fqm \to \Fqm$ that satisfies both $\der(a + b) = \der(a) + \der(b)$ and
	$\der(a \cdot b) = \der(a) \cdot b + \aut(a) \cdot \der(b)$ for all $a, b \in \Fqm$.
	In our finite-field setting, every $\aut$-derivation is an inner derivation, that is
	$\der = \gamma (\Id - \aut)$ for a $\gamma \in \Fqm$ and the identity $\Id$ on $\Fqm$.
	When the automorphism $\aut$ is clear from the context, we often write $\der_{\gamma}$ to refer to the derivation
	corresponding to $\gamma \in \Fqm$.

	For a fixed pair $(\aut, \der)$, we can group the elements of $\Fqm$ with respect to an equivalence relation called $(\aut, \der)$-conjugacy:

	\begin{definition}
		Two elements $a, b \in \Fqm$ are called \defemph{$(\aut, \der)$-conjugate} if there is a $c \in \Fqm^{\ast}$ with
		\begin{equation}
			\conj{a}{c} \defeq \aut(c) a c^{-1} + \der(c) c^{-1} = b.
		\end{equation}
		All conjugates of $a \in \Fqm$ are collected in the respective \defemph{conjugacy class}
		\begin{equation}
			\set{K}(a) \defeq \{ a^c : c \in \Fqm^{\ast} \} \subseteq \Fqm.
		\end{equation}
		For $\der = \der_{\gamma}$ with $\gamma \in \Fqm$, the class $\textsc{\set{K}}(\gamma)$ is called trivial conjugacy class.
	\end{definition}

	\subsection{Isometries in the Sum-Rank Metric}

	As most code-based cryptosystems use isometric disguising, we quickly recall the characterization of sum-rank isometries.
	Note that we have to differentiate between $\Fq$-linear and $\Fqm$-linear isometries.
	The former were studied in~\cite[Prop.~4.26]{Neri2022}, whereas the latter were considered
	in~\cite{Mar2020a,AlfLobNer2022}.
	Precisely, the special case of equal block lengths (i.e., $\n =
	\left( \frac{\len}{\shots}, \dots, \frac{\len}{\shots} \right)$) was treated in~\cite[Thm.~2]{Mar2020a} and the
	generalization to arbitrary block lengths and the extension to semilinear isometries is due
	to~\cite[Sec.~3.3]{AlfLobNer2022}.
	We focus on $\Fqm$-(semi)linear isometries because of our motivation from code-based cryptography.
	Namely, we consider the following:

	\begin{definition}
		A bijective map $\iota: \Fqm^{n} \to \Fqm^{n}$ is a \defemph{(sum-rank) isometry} on $\Fqm^{n}$ if it is
		sum-rank preserving, that is if $\SumRankDist(\x) = \SumRankDist(\iota(\x))$ holds for all $\x \in \Fqm^{n}$.
		We call an isometry \defemph{linear} when it is $\Fqm$-linear.
		A \defemph{semilinear} isometry $\iota$ is additive and there exists an $\Fqm$-automorphism $\theta$ such that
		$\iota$ fulfills $\iota(a\x) = \theta(a) \iota(\x)$ for all $a \in \Fqm$ and all $\x \in \Fqm^{n}$.
	\end{definition}

	Recall that the \defemph{general linear group} $\GL(n, \Fq)$ contains all full-rank matrices of size $n \times n$
	over $\Fq$ and that the \defemph{symmetric group} $\Sym_{n}$ consists of all permutations of $n$ elements.
	We introduce the notations
	\begin{align}
		\GL(\n, \Fq) &\defeq \GL(\lenShot{1}, \Fq) \times \dots \times \GL(\lenShot{\shots}, \Fq) \\
		\text{and} \quad
		\Sym_{\lambda(\n)} &\defeq \Sym_{\lambda_1} \times \dots \times \Sym_{\lambda_{k}},
	\end{align}
	where $\lambda(\n)$ counts the occurrences of distinct entries of $\n$ (see~\autoref{subsec:sum-rank_metric}).
	Note that $\Sym_{\lambda(\n)}$ is a subgroup of $\Sym_{\sum_j \lambda_j} = \Sym_{\shots}$.

	\begin{theorem}[Sum-Rank Isometries~\cite{Mar2020a,AlfLobNer2022}]\label{thm:sum-rank_isometries}
		The group of $\Fqm$-linear isometries on $\Fqm^{n}$ is
		\begin{equation}
			\linIsometries \defeq \left( (\Fqm^{\ast})^{\shots} \times \GL(\n, \Fq) \right) \rtimes \Sym_{\lambda(\n)}.
		\end{equation}
		Its action $\actLin: \linIsometries \times \Fqm^{n} \to \Fqm^{n}$ is defined as
		\begin{equation}
			\actLin(\iota, \x) \defeq
			\left( c_1 \shot{\x}{\pi^{-1}(1)} \M_1 \mid \dots \mid
			c_{\shots} \shot{\x}{\pi^{-1}(\shots)} \M_{\shots} \right)
		\end{equation}
		for $\iota = \left( (c_1, \dots, c_{\shots}), (\M_1, \dots, \M_{\shots}), \pi \right)$ and
		$\x \in \Fqm^{n}$.
		Similarly, the group of $\Fqm$-semilinear isometries on $\Fqm^{n}$ is
		\begin{align}
			\semilinIsometries &\defeq \linIsometries \rtimes \Aut(\Fqm)
		\end{align}
		and its action $\actSemilin: \semilinIsometries \times \Fqm^{n} \to \Fqm^{n}$ is given by
		\begin{equation}
			\actSemilin((\iota, \theta), \x) \defeq \theta(\actLin(\iota, \x))
		\end{equation}
		for $(\iota, \theta) \in \semilinIsometries$ and $\x \in \Fqm^{n}$.
	\end{theorem}

	Since MacWilliams' extension theorem does not hold in this general setting (see~\cite[Ex.~2.9~(a)]{barra2015}
	for a counterexample in the rank-metric case), code equivalence in the sum-rank
	metric is defined by means of isometries of the whole space (cp.~\cite[Def.~3.9]{AlfLobNer2022}).

	\begin{definition}
		Two sum-rank-metric codes $\mycode{C}, \mycode{D} \subseteq \Fqm^{n}$ are called
		\defemph{linearly equivalent} if there is a linear isometry $\iota \in \linIsometries$ such that
		\begin{equation}
			\actLin(\iota, \mycode{C}) \defeq \{ \actLin(\iota, \c) : \c \in \mycode{C} \} = \mycode{D}.
		\end{equation}
		They are \defemph{semilinearly equivalent} if there is $(\iota, \theta) \in \semilinIsometries$ such that
		\begin{equation}
			\actSemilin((\iota, \theta), \mycode{C}) \defeq \{ \actSemilin((\iota, \theta), \c) : \c \in \mycode{C} \}
			= \mycode{D}.
		\end{equation}
	\end{definition}

	\subsection{Skew Polynomials}

	The \defemph{skew-polynomial ring} $\SkewPolyring$ is defined as the set of polynomials $f(x) = \sum_i f_i x^i$
	with finitely many nonzero coefficients $f_i \in \Fqm$.
	It is equipped with conventional polynomial addition but the multiplication is determined by the rule
	$xa = \aut(a)x + \der(a)$.
	Similar to conventional polynomial rings, we define the \defemph{degree} of a nonzero skew polynomial
	$f(x) = \sum_i f_i x^i \in \SkewPolyring$ as $\deg(f) \defeq \max \{i : f_i \neq 0\}$ and set the degree of the zero
	polynomial to $-\infty$.

	Note that despite lots of similarities to $\Polyring$, the same evaluation strategy (i.e., $f(c) = \sum_i f_i c^i$ for
	$c \in \Fqm$) does not work in this setting.
	Instead, the literature provides two different ways to adequately evaluate skew polynomials: remainder evaluation
	and generalized operator evaluation.
	We will focus on the latter in this work.

	For $a, b \in \Fqm$, define the operator
	\begin{equation}
		\op{a}{b} \defeq \aut(b) a + \der(b)
	\end{equation}
	and its powers $\opexp{a}{b}{i} \defeq \op{a}{\opexp{a}{b}{i-1}}$ for $i \geq 0$ (with $\opexp{a}{b}{0} = b$ and
	$\opexp{a}{b}{1} = \op{a}{b}$).
	For $\a = (a_1, \dots, a_{\shots}) \in \Fqm^{\shots}$, and $\B \in \Fqm^{k \times n}$, we write
	$\op{\a}{\B} \defeq (\op{a_1}{\shot{\B}{1}} \mid \dots \mid \op{a_{\shots}}{\shot{\B}{\shots}})$, where
	$\op{a_i}{\shot{\B}{i}}$ stands for the elementwise application of $\op{a_i}{\cdot}$ to the entries of
	$\shot{\B}{i}$ for $1 \leq i \leq \shots$.
	This notation also applies to vectors $\b \in \Fqm^{n}$ and can be extended to powers of the operator.

	In the zero-derivation case, the $i$-fold application of the above defined operator can be expressed as
	\begin{equation}
	    \opexp{a}{b}{i} = \aut^{i}(b) \cdot \genNorm{i}{a}
	\end{equation}
	for any $a, b \in \Fqm$ and $i \in \NN^{\ast}$.
	Here, $\genNorm{i}{a} \defeq \prod_{j=0}^{i-1} \aut^{j}(a) = \aut^{i-1}(a) \dots \aut(a) \cdot a$ denotes
	the \emph{generalized power function}.

	\begin{lemma}\label{lem:opProductRule}
		The equality $\op{a}{bc} = \aut(b) \op{a}{c} + \der(b) c$ holds for any $a, b, c \in \Fqm$.
	\end{lemma}
	
	\begin{proof}
		The definition of $\op{a}{\cdot}$ and the product rule for derivations yield
	    \begin{align}
			\op{a}{bc} &= \aut(bc) a + \der(bc) = \aut(bc) a + \der(b)c + \aut(b)\der(c) \\
			&= \aut(b)(\aut(c)a + \der(c)) + \der(b)c = \aut(b) \op{a}{c} + \der(b)c.
		\end{align}
	\qed
	\end{proof}

	Let us now define the generalized operator evaluation of skew polynomials:

	\begin{definition}
		The \defemph{generalized operator evaluation} of a skew polynomial $f(x) = \sum_i f_i x^i \in \SkewPolyring$ at
		a point $b \in \Fqm$ with respect to an evaluation parameter $a \in \Fqm$ is given by
		\begin{equation}
			\opev{f}{b}{a} \defeq \sum_i f_i \opexp{a}{b}{i}.
		\end{equation}
	\end{definition}

	For a vector $\x \in \Fqm^{n}$, a vector $\a = (a_1, \ldots, a_\shots) \in \Fqm^{\shots}$, and a parameter
	$d \in \NN^{\ast}$, the \emph{generalized Moore matrix} $\opMoore{d}{\x}{\a}$ is defined as
	\begin{align}\label{eq:def_gen_moore_mat}
		\opMoore{d}{\x}{\a} &\defeq
		\left( \opVandermonde{d}{\x^{(1)}}{a_1}, \dots, \opVandermonde{d}{\x^{(\shots)}}{a_\shots} \right)
		\in \Fqm^{d \times n},
		\\
		\text{where }
		\opVandermonde{d}{\x^{(i)}}{a_i} &\defeq
		\begin{pmatrix}
			x^{(i)}_1 & \cdots & x^{(i)}_{\lenShot{i}}
			\\
			\op{a_i}{x^{(i)}_1} & \cdots & \op{a_i}{x^{(i)}_{\lenShot{i}}}
			\\
			\vdots & \ddots & \vdots
			\\
			\opexp{a_i}{x^{(i)}_1}{d-1} & \cdots & \opexp{a_i}{x^{(i)}_{\lenShot{i}}}{d-1}
		\end{pmatrix}
		\quad \text{for } 1 \leq i \leq \shots.
	\end{align}
	If $\a$ contains representatives of pairwise distinct nontrivial conjugacy classes of $\Fqm$ and $\rk_{q}\left(\x^{(i)}\right) =
	\lenShot{i}$ for all $1 \leq i \leq \shots$, we have by~\cite[Thm.~2]{Mar2018}
	and~\cite[Thm.~4.5]{lam1988vandermonde} that $\rk_{q^m}\left(\opMoore{d}{\x}{\a}\right) = \min(d, \len)$.

	\subsection{(Generalized) Linearized Reed--Solomon Codes}

	Let us recall the definition of \ac{LRS} codes that generalize both \ac{RS} and Gabidulin codes.
	\ac{LRS} codes are evaluation codes with respect to skew polynomials, which specialize to conventional and
	linearized polynomials in the Hamming- and the rank-metric setting, respectively.

	\begin{definition}[Linearized Reed--Solomon Codes~{\cite[Def.~31]{Mar2018}}]\label{def:lrs_code}
		Let	$\a = (a_1, \dots, a_{\shots}) \in \Fqm^{\shots}$ consist of representatives of distinct nontrivial
		conjugacy classes of $\Fqm$.
		Choose a vector $\vecbeta \in \Fqm^{n}$ whose blocks $\vecbeta^{(i)} =
		\left( \shot{\beta}{i}_1, \dots, \shot{\beta}{i}_{\lenShot{i}} \right)$ contain $\Fq$-linearly independent elements for all
		$i = 1, \dots, \shots$.
		Then, a \defemph{\acf{LRS} code} of length $\len$ and dimension $k$ is defined as
		\begin{equation}
			\linRS{\vecbeta, \a; \n, k} \defeq
			\left\{ \left( \shot{\c}{1}(f) \mid \dots \mid \shot{\c}{\shots}(f) \right):
			f \in \SkewPolyring_{<k} \right\}
			\subseteq \Fqm^{n}
		\end{equation}
		where $\shot{\c}{i}(f) \defeq
		\left( \opev{f}{\shot{\beta}{i}_1}{a_i}, \ldots, \opev{f}{\shot{\beta}{i}_{\lenShot{i}}}{a_i} \right)$.
	\end{definition}

	Note that \ac{LRS} codes reach the Singleton-like bound $d \leq n - k + 1$ from~\cite[Prop.~34]{Mar2018} with
	equality, where $d$ denotes the minimum sum-rank distance of the code.
	They are thus \emph{\ac{MSRD}} codes.

	The generalized Moore matrix $\opMoore{k}{\vecbeta}{\a}$ is a generator matrix of the code $\linRS{\vecbeta, \a; \n, k}$.
	Since a generator matrix of this form is desirable as it e.g.\ gives rise to known efficient decoding algorithms, we call it a
	\emph{canonical} generator matrix of $\linRS{\vecbeta, \a; \n, k}$.
	Note that the parameters $\vecbeta$ and $\a$ of a canonical generator matrix are in general not uniquely determined,
	and not even fixing a particular $\a$ ensures the uniqueness of $\vecbeta$.

	In the zero-derivation case, the dual of an \ac{LRS} code can be described as
	\begin{equation}\label{eq:lrs_dual_zero_der}
		\linRS{\vecbeta,\a;\n,k}^\perp = \linRS{\vecalpha,\autinv(\a);\n,n-k}_{\autinv},
	\end{equation}
	where the index $\autinv$ on the right-hand side stands for the fact that it is an \ac{LRS} code with respect to
	the inverse automorphism $\autinv$ (see~\cite{caruso2019residues,Caruso2022duals}).
	The vector $\vecalpha=(\vecalpha^{(1)}\mid\dots\mid\vecalpha^{(\shots)})\in\Fqm^n$ satisfies
	\begin{equation}\label{eq:syndrome_eqs_lrs}
    	\sum_{i=1}^{\shots}\sum_{j=1}^{n_i}\alpha_{j}^{(i)}\opexp{a_i}{\beta_{j}^{(i)}}{h-1}=0
    	\quad \text{for all } h=1,\dots,n-1
	\end{equation}
	and has sum-rank weight $\SumRankWeight(\vecalpha)=n$ according to~\cite[Thm.~4]{martinez2019reliable}.
	In particular, the dual of a zero-derivation \ac{LRS} code is again an \ac{LRS} code.
	When nonzero derivations are allowed, the duals of \ac{LRS} codes are linearized Goppa codes which are
	(noncanonically) isomorphic to \ac{LRS} codes~\cite{caruso2019residues,Caruso2022duals}.

	As the proof of~\autoref{thm:GLRS_semilinear_equivalence} shows, codes that are (semi)linearly
	equivalent to \ac{LRS} codes are not necessarily \ac{LRS} codes themselves.
	However, this is true for a more general code family that is obtained by allowing nonzero block multipliers.
	We define \ac{GLRS} codes as follows:

	\begin{definition}[Generalized Linearized Reed--Solomon Codes]
		Let $\mycode{C} \defeq \linRS{\vecbeta, \a; \n, k}$ be an \ac{LRS} code as in~\autoref{def:lrs_code}.
		Further, let $\v = \left( v_1, \dots, v_{\shots} \right) \in \Fqm^{\shots}$ be a vector of nonzero $\Fqm$-elements.
		We define the \defemph{\acl{GLRS} code} $\GLRS{\vecbeta, \a, \v; \n, k}$ as
		\begin{equation}
			\GLRS{\vecbeta, \a, \v; \n, k} \defeq
			\left\{ \left( v_1 \shot{\c}{1} \mid \dots \mid v_{\shots} \shot{\c}{\shots} \right):
			\c \in \mycode{C} \right\}
			\subseteq \Fqm^{n}.
		\end{equation}
	\end{definition}
	Remark that we recover \ac{LRS} codes from \ac{GLRS} codes for $\v$ being the all-one vector.
	Since multiplying blocks with different nonzero $\Fqm$-elements is a sum-rank isometry according
	to~\autoref{thm:sum-rank_isometries}, we obtain the following corollary:

	\begin{corollary}
	    The minimum sum-rank distance of the code $\GLRS{\vecbeta, \a, \v; \n, k}$ is $d = n - k + 1$.
		Therefore, \ac{GLRS} codes are \ac{MSRD}.
	\end{corollary}

	The code $\GLRS{\vecbeta, \a, \v; \n, k}$ has a generator matrix of the form
	\begin{equation}\label{eq:GLRS_gen_mat}
		\G = \left( v_1 \opVandermonde{k}{\shot{\vecbeta}{1}}{a_1} \mid \dots \mid v_{\shots} \opVandermonde{k}{\shot{\vecbeta}{\shots}}{a_{\shots}} \right).
	\end{equation}
	Similar to the \ac{LRS} case, we call any generator matrix of this form a \emph{canonical} generator matrix of $\GLRS{\vecbeta, \a, \v; \n, k}$.
	Note that a canonical generator matrix of a \ac{GLRS} code depends not only
	on the parameters $\vecbeta$ and $\a$ but also on the block multipliers $\v$.

	\section{Problem Statement}\label{sec:problem_statement}

	The main problem we want to solve is distinguishing \ac{GLRS} codes, that were disguised by means of
	$\Fqm$-semilinear isometries, from random sum-rank-metric codes of the same length and dimension.
	Formally, we state this task as follows:

	\begin{problem}[Distinguishing \ac{GLRS} Codes up to Semilinear Equivalence]\label{prob:distinguish_semilinear_equivalence}
		Given a full-rank matrix $\M \in \Fqm^{k \times n}$, decide if there are parameters $\vecbeta \in \Fqm^{n}$,
		$\a \in \Fqm^{\shots}$, $\v \in \Fqm^{\shots}$, $\aut \in \Gal(\Fqm/\Fq)$, and $\der$ being a $\aut$-derivation,
		such that $\langle \M \rangle$ is semilinearly equivalent to $\GLRS{\vecbeta, \a, \v; \n, k}$.
	\end{problem}

	We now investigate how $\Fqm$-semilinear transformations affect \ac{GLRS} codes to get a better understanding of the
	problem.
	\autoref{thm:GLRS_semilinear_equivalence} shows that every semilinear isometry (cp.~\autoref{thm:sum-rank_isometries})
	transforms a \ac{GLRS} code into another \ac{GLRS} code with possibly different parameters:

	\begin{theorem}\label{thm:GLRS_semilinear_equivalence}
		Let $\mycode{C} = \GLRS{\vecbeta, \a, \v; \n, k}$ be a \ac{GLRS} code with respect to $\aut$ and $\der \defeq \der_{\gamma}$.
		Let further $\iota \in \linIsometries$ denote an
		$\Fqm$-linear isometry with $\iota = \left( (c_1, \dots, c_{\shots}), (\M_1,\dots, \M_{\shots}), \pi \right)$.
		Then, the linearly equivalent code $\hat{\mycode{C}} \defeq \actLin(\iota, \mycode{C})$ is also a \ac{GLRS} code with respect to $\aut$ and $\der$.
		Namely, $\hat{\mycode{C}} = \GLRS{\hat{\vecbeta}, \hat{\a}, \hat{\v}; \n, k}$ with
		$\hat{\vecbeta} = (\shot{\vecbeta}{\pi^{-1}(1)} \M_1 \mid \dots \mid \shot{\vecbeta}{\pi^{-1}(\shots)} \M_{\shots})$,
		$\hat{\a} = (a_{\pi^{-1}(1)}, \dots, a_{\pi^{-1}(\shots)})$, and
		$\hat{\v} = (c_1 v_{\pi^{-1}(1)}, \dots, c_{\shots} v_{\pi^{-1}(\shots)})$.

		For a semilinear isometry $(\iota, \theta) \in \semilinIsometries$ with $\iota$ as above and $\theta \in \Aut(\Fqm)$,
		the code $\actSemilin((\iota, \theta), \mycode{C})$ is  a \ac{GLRS} code with
		respect to the automorphism $\aut$ and the possibly different derivation $\der_{\theta(\gamma)} \defeq \theta(\gamma) (\Id - \aut)$.
		Its parameters are $\theta(\hat{\vecbeta})$, $\theta(\hat{\a})$, and $\theta(\hat{\v})$, where $\theta$ is
		applied elementwise to the vectors.
	\end{theorem}

	\begin{proof}
		Let us use the shorthand notations $v_{\pi_i} \defeq v_{\pi^{-1}(i)}$, $a_{\pi_i} \defeq a_{\pi^{-1}(i)}$, and
		$\shot{\vecbeta}{\pi_i} \defeq \shot{\vecbeta}{\pi^{-1}(i)}$ throughout this proof.
		$\mycode{C}$ has a generator matrix of the form $\G \defeq \left( v_1 \opVandermonde{k}{\vecbeta^{(1)}}{a_1}, \dots,
		v_{\shots} \opVandermonde{k}{\vecbeta^{(\shots)}}{a_{\shots}} \right)$.
		If $\iota$ acts on the $j$-th row of $\G$ for $j \in \{1, \dots, k\}$, we obtain
		\begin{equation}\label{eq:linIsoOnMooreRow}
			\left( c_1 v_{\pi_1} \opexp{a_{\pi_1}}{\shot{\vecbeta}{\pi_1}}{j-1} \M_1 \mid \dots \mid
			c_{\shots} v_{\pi_{\shots}} \opexp{a_{\pi_{\shots}}}{\shot{\vecbeta}{\pi_{\shots}}}{j-1} \M_{\shots} \right).
		\end{equation}
		Since generalized operator evaluation is $\Fq$-linear, we get
		$\opexp{a_{\pi_{i}}}{\shot{\vecbeta}{\pi_{i}}}{j-1} \M_{i} = \opexp{a_{\pi_{i}}}{\shot{\vecbeta}{\pi_{i}} \M_{i}}{j-1}$
		for all $i = 1, \dots, \shots$
		and thus,~\eqref{eq:linIsoOnMooreRow} is exactly the $j$-th row of
		\begin{equation}
			\hat{\G} \defeq \left( c_1 v_{\pi_1} \opVandermonde{k}{\vecbeta^{(\pi_1)} \M_1}{a_{\pi_1}}, \dots,
			c_{\shots} v_{\pi_{\shots}} \opVandermonde{k}{\vecbeta^{(\pi_{\shots})} \M_{\shots}}{a_{\pi_{\shots}}} \right).
		\end{equation}
		As $\hat{\G}$ generates $\hat{\mycode{C}}$, this proves the first part of the theorem.
		The second one follows from the observation
		\begin{equation}\label{eq:semilinInductionStatement}
		    \theta(v \opexp{a}{\beta}{j-1}) = \theta(v) \opfullparamexp{\theta(a)}{\theta(\beta)}{\aut, \der_{\theta(\gamma)}}{j-1}
		\end{equation}
		for any $v, a, \beta \in \Fqm^{\ast}$ and $j \in \NN^{\ast}$ with
		$\opfullparam{\cdot}{\cdot}{\aut, \der_{\theta(\gamma)}}$ denoting the generalized operator evaluation with
		respect to the automorphism $\aut$ and the derivation $\der_{\theta(\gamma)} \defeq \theta(\gamma) (\Id - \aut)$.
		\eqref{eq:semilinInductionStatement} can be verified by induction over $j$.
	\qed
	\end{proof}

	In fact, this shows that \ac{GLRS} codes with respect to a fixed automorphism and a fixed derivation are closed
	under linear equivalence.
	If we allow different derivations for a fixed automorphism, \ac{GLRS} codes are even closed under semilinear
	equivalence.
	This means, intuitively speaking, that~\autoref{prob:distinguish_semilinear_equivalence} boils down to distinguishing \ac{GLRS}
	codes.
	We hence formulate and focus on~\autoref{prob:distinguish}:

	\begin{problem}[Distinguishing \ac{GLRS} Codes]\label{prob:distinguish}
	    Given a full-rank matrix $\M \in \Fqm^{k \times n}$, decide if there are parameters $\vecbeta \in \Fqm^{n}$,
		$\a \in \Fqm^{\shots}$, $\v \in \Fqm^{\shots}$, $\aut \in \Gal(\Fqm/\Fq)$, and $\der$ being a $\aut$-derivation,
		such that $\langle \M \rangle = \GLRS{\vecbeta, \a, \v; \n, k}$.
	\end{problem}

	Let us describe more precisely how the two above-defined problems are related in case we assume the knowledge of the automorphism $\aut$ and the derivation
	$\der \defeq \der_{\gamma}$.
	If we restrict ourselves to linear equivalence,~\autoref{prob:distinguish_semilinear_equivalence} is equivalent to~\autoref{prob:distinguish}
	since every code that is linearly equivalent to a \ac{GLRS} code with respect to $\aut$ and $\der$ is a \ac{GLRS} code
	with respect to the same automorphism and derivation.
	In the more general, semilinear setting, we can solve~\autoref{prob:distinguish_semilinear_equivalence} by solving
	multiple instances of~\autoref{prob:distinguish}.
	Namely, we have to consider~\autoref{prob:distinguish} for all derivations $\der_{\theta(\gamma)} \defeq \theta(\gamma)(\Id - \aut)$
	with $\theta \in \Aut(\Fqm)$ according to~\autoref{thm:GLRS_semilinear_equivalence}.
	As $\vert \Aut(\Fqm) \vert = sm$ for $s$ being the extension degree of $\Fq$ over its prime field, we obtain that~\autoref{prob:distinguish_semilinear_equivalence}
	is equivalent to $sm$ instances of~\autoref{prob:distinguish}.

	We present two polynomial-time distinguishers that partly solve~\autoref{prob:distinguish} when $\aut$ and
	$\der$ are known in~\autoref{sec:distinguishers}.
	However, the pure knowledge whether a matrix generates a \ac{GLRS} code or not does not yet break a hypothetical McEliece-like cryptosystem based on \ac{GLRS} codes.
	We rather wish to recover an efficient decoding algorithm for the publicly known code by e.g.\ finding a
	canonical generator matrix.
	Therefore, the following problem is of great interest:

	\begin{problem}[Recovering a Canonical GLRS Generator Matrix]\label{prob:recover}
	    Given an arbitrary generator matrix $\G \in \Fqm^{k \times n}$ of a \ac{GLRS} code $\mycode{C}$,
		find parameters $\vecbeta \in \Fqm^{n}$, $\a \in \Fqm^{\shots}$, $\v \in \Fqm^{\shots}$, $\aut \in \Gal(\Fqm/\Fq)$, and $\der$ being a $\aut$-derivation,
		such that $\left( v_1 \opVandermonde{k}{\shot{\vecbeta}{1}}{a_1} \mid \dots \mid v_{\shots}
		\opVandermonde{k}{\shot{\vecbeta}{\shots}}{a_{\shots}} \right)$
		is a canonical generator matrix of $\mycode{C}$.
	\end{problem}

	We study~\autoref{prob:recover} in~\autoref{sec:decoding} and show two techniques to partially solve it for \ac{GLRS} codes in the zero-derivation case
	with known automorphism $\aut$.

	\section{Distinguishers for \ac{GLRS} Codes}\label{sec:distinguishers}

	This section contains two approaches that solve~\autoref{prob:distinguish}, that is the task of distinguishing
	\ac{GLRS} codes from random codes, for many instances.
	In both cases, we assume the knowledge of the automorphism $\aut$ and the derivation $\der$ with respect to which
	the code should be distinguished.

	In~\autoref{sec:square_distinguisher}, we focus on a square-code distinguisher that is inspired by an \ac{RS}-code
	distinguisher.
	It works for \ac{GLRS} codes constructed by means of the identity automorphism and zero derivation.

	Afterwards, we present an Overbeck-like distinguisher inspired by the rank-metric case
	in~\autoref{sec:overbeck_distinguisher}.
	This approach can handle any valid combination of automorphism and derivation but requires the knowledge of the
	evaluation-parameter vector $\a$.
	Moreover, the Overbeck-type distinguisher cannot deal with block multipliers and is thus applicable
	to \ac{LRS} codes only.
	However, \ac{GLRS} codes can still be handled by applying the distinguisher at most $(q^m-1)^{\shots}$ times (see~\autoref{sec:overbeck_distinguisher}
	for more details).

	We experimentally verified all results presented in this section for different parameter sets with an implementation in~SageMath~\cite{sage}.

	\subsection{A Square-Code Distinguisher}\label{sec:square_distinguisher}

	The first polynomial-time attack on a McEliece/Niederreiter variant based on \ac{GRS} codes was proposed by \acl{SaS} in~\cite{sidelnikov1992insecurity}.
	The attack was later on refined by Wieschebrink to attack the improved Berger--Loidreau cryptosystem~\cite{wieschebrink2006attack}, which is based on \ac{GRS} subcodes.
	The approach from~\cite{wieschebrink2006attack} was further improved in~\cite{wieschebrink2010cryptanalysis} to work with smaller subcodes and thus to break the cryptosystem for most practical parameters.
	The attack in~\cite{wieschebrink2010cryptanalysis} is based on the properties of the elementwise product (or \emph{Schur-square}) of a code.
	For any vectors $\x, \y \in \Fqm^n$ we define the \defemph{elementwise product} (also referred to as \emph{Schur} or \emph{star} product) of $\x$ and $\y$ as
	\begin{equation*}
		\x \star \y \defeq (x_1 y_1, x_2 y_2. \dots, x_n y_n).
	\end{equation*}
	The \defemph{square-code} of an $\Fqm$-linear code $\mycode{C} \subseteq \Fqm^n$ is defined as
	\begin{equation*}
		\mycode{C} \star \mycode{C} \defeq \left\{ \c_1 \star \c_2: \c_1, \c_2 \in \mycode{C} \right\}.
	\end{equation*}

	The main observation for distinguishing a random linear code in $\Fqm^{n}$ from a \ac{GRS} code $\mycode{C}$ is that the squared \ac{GRS} code has dimension $\dim(\mycode{C} \star \mycode{C}) = \min(n, \, 2k-1)$, which is small compared to the expected dimension of a squared random linear code.
	Note that a similar technique was used for the power decoding of \ac{RS} codes beyond the unique-decoding radius (see~\cite[Lemma~1]{schmidt2006decoding}).
	
	We will now derive a similar distinguisher for \ac{GLRS} codes constructed from skew-polynomial rings with identity automorphism $\aut = \Id$. Observe that in this case the only possible derivation is the zero derivation.
	\autoref{lem:lrs_id_aut} provides some basic results required for deriving a square-code distinguisher for \ac{GLRS} codes:

	\begin{lemma}\label{lem:lrs_id_aut}
		For $\aut = \Id$, let $\mycode{C}=\GLRS{\vecbeta, \a, \v; \n, k}$ be a \ac{GLRS} code constructed by polynomials from $\SkewPolyringZeroDer_{<k}=\Polyring_{<k}$.
		Then we have that
		\begin{align}
			\mycode{C} = \bigl\{ &\bigl(f(a_1), \dots, f(a_1) \mid \dots \mid f(a_\shots), \dots, f(a_\shots)\bigr) \cdot \diag\bigl((v_1 \shot{\vecbeta}{1} \mid \dots \mid v_{\shots} \shot{\vecbeta}{\shots})\bigr)\\
			&: f \in \Polyring_{<k}\bigr\},
		\end{align}
		where $f(\cdot)$ denotes ordinary polynomial evaluation.
	\end{lemma}

	\begin{proof}
		Since $\aut$ is the identity automorphism, the generalized operator evaluation of $f \in \SkewPolyringZeroDer$ at an element $\beta_j^{(i)} \in \Fqm$ with respect to the evaluation parameter $a_i \in \Fqm$ is
		\begin{align*}
			\opev{f}{\beta_j^{(i)}}{a_i} 
			&= \sum_{l=0}^{k-1}f_l\opexp{a_i}{\beta_j^{(i)}}{l}
			= \sum_{l=0}^{k-1}f_l\aut^{l}(\beta_j^{(i)})\genNorm{l}{a_i}
			=\beta_j^{(i)} \sum_{l=0}^{k-1} f_l a_i^l
			=\beta_j^{(i)} f(a_i),
		\end{align*}
		where $f(\cdot)$ denotes the ordinary polynomial evaluation.
		Hence, any $\c \in \mycode{C}$ can be written as
		\begin{align*}
			\c &= (v_1 \opev{f}{\beta_1^{(1)}}{a_1}, \dots, v_1 \opev{f}{\beta_{n_1}^{(1)}}{a_1} \mid \dots \mid v_{\shots} \opev{f}{\beta_1^{(\shots)}}{a_\shots}, \dots, v_{\shots} \opev{f}{\beta_{n_\shots}^{(\shots)}}{a_\shots})
			\\ 
			&=(v_1 \beta_1^{(1)}f(a_1), \dots, v_1 \beta_{n_1}^{(1)}f(a_1) \mid \dots \mid v_{\shots} \beta_1^{(\shots)}f(a_\shots), \dots, v_{\shots} \beta_{n_\shots}^{(\shots)}f(a_\shots)).
		\end{align*}
	\qed
	\end{proof}

	This allows the derivation of~\autoref{lem:square_code_dim} which is a result about the dimension of the square code
	of \ac{GLRS} codes and, in contrast, of random linear codes.

	\begin{lemma}\label{lem:square_code_dim}
	\begin{enumerate}
	\item Let $\mycode{C}\subseteq \Fqm^n$ be a \ac{GLRS} code of dimension $k$ with respect to $\aut=\Id$. Then
	\begin{equation}
	    \dim(\mycode{C} \star \mycode{C}) =  \min(\ell, 2k-1).
	\end{equation}
	\item Let $\mycode{C}\subseteq \Fqm^n$ be a linear code of dimension $k$ that was chosen uniformly at random. Then
	\begin{equation}
	    \Pr\left( \dim(\mycode{C} \star \mycode{C})<\min\left(n, \frac{k(k+1)}{2}\right) \right) \xrightarrow{k\to \infty} 0,
	\end{equation}
	where $\Pr(\cdot)$ denotes the probability of the event in parentheses.
	\end{enumerate}
	\end{lemma}
	
	\begin{proof}
	\begin{enumerate}
	\item Let $\c, \c'$ be two codewords from $\mycode{C} \star \mycode{C}$ constructed by the evaluation of the polynomials $f, g \in \SkewPolyringZeroDer$ having the maximal degree $\deg(f)=\deg(g)=k-1$.
	Then, by~\autoref{lem:lrs_id_aut}, we have that $\c \star \c'$ has the form
	\begin{align*}
		\c \star \c' = &\left((f \cdot g)(a_1), \dots, (f \cdot g)(a_1) \mid  \dots \mid (f \cdot g)(a_\shots), \dots, (f \cdot g)(a_\shots)\right)\\
		&\cdot \diag\left(\left(v_1^2 \left({\shot{\vecbeta}{1}}\right)^2 \mid \dots \mid v_{\shots}^2 \left({\shot{\vecbeta}{\shots}}\right)^2\right)\right),
	\end{align*}
	where the squaring of the blocks $\shot{\vecbeta}{i}$ for $i = 1, \dots, \shots$ is understood elementwise.
	Since $\a$ contains representatives of different conjugacy classes of $\Fqm$, the elements in $\a$ are pairwise distinct.
	Since $\vecbeta$ contains block-wise $\Fq$-linearly independent elements, all entries in $\vecbeta$ are nonzero.
	Together with the fact that $\v$ contains only nonzero elements this implies that the diagonal matrix has full rank $n$.
	Hence, by considering only the first column of each block, we get a~\ac{GRS} code of length $\shots$ and dimension $\deg(f \cdot g) + 1 = 2k - 1$.
	The size of the corresponding generator matrix is $(2k-1)\times \ell$,
	 which yields the statement.
	\item This follows directly from~\cite{wieschebrink2006attack}.
	\end{enumerate}
\qed
	\end{proof}

\autoref{thm:square_code_dist} summarizes the results for the Wieschebrink-like square-code distinguisher for \ac{GLRS} codes in the identity-automorphism case.

\begin{theorem}[Square-Code Distinguisher]\label{thm:square_code_dist}
	Let $2 < k\leq \frac{n}{2}$ and let $\aut$ be the identity automorphism.
	Given a generator matrix of a $k$-dimensional code in $\Fqm^n$, we can distinguish a \ac{GLRS} code from a random code with high probability\footnote{In fact, the distinguisher recognizes a
	\ac{GLRS} code with probability one. But, with a small probability, it might wrongly declare a non-\ac{GLRS} code to be a \ac{GLRS} code.} in $\OCompl{n^5}$ operations in $\Fqm$.
\end{theorem}

\begin{proof}
	Using \autoref{lem:square_code_dim} we can distinguish a \ac{GLRS} code with high probability from a random linear code by considering the dimension of the square code.
	The complexity, which is in the order of
	\begin{equation*}
		\OCompl{k^4n + k^2n + k^2 (n-k)^2n} \subseteq \OCompl{n^5}
	\end{equation*}
	operations in $\Fqm$, follows from~\cite{wieschebrink2010cryptanalysis}.
\qed
\end{proof}

	\subsection{An Overbeck-like Distinguisher}\label{sec:overbeck_distinguisher}

	Overbeck proposed a distinguisher for Gabidulin codes in~\cite{Ove2005,Ove2007,Ove2007PhD}.
	The main idea is to repeatedly apply the Frobenius automorphism to the public generator matrix and stack the results
	vertically.
	Since there is a generator matrix of a Gabidulin code whose $i$-th row is the $(i-1)$-fold application of the Frobenius
	automorphism to a generating vector, the rank of the stacked matrix will only increase by one for each new matrix
	block.
	But random full-rank matrices behave differently and the stacked matrix has much higher rank in general.

	\acl{HaMaR}~\cite{horlemann2018extension} used a slightly different approach, which we will call \emph{\acs{HMR} approach}
	for short, to recover the secret parameters of a Gabidulin code.
	We mention their technique because it gives rise to a distinguisher and it is similar to Overbeck's approach,
	as it also makes use of the repeated application of the Frobenius automorphism to the public generator matrix.
	But instead of considering the sum of the corresponding codes, the \acs{HMR} approach focuses on the
	intersection of the codes and shows that its dimension only decreases by one for each iteration step.
	Again, random codes show a different behavior under this operation.

	We now present a generalization of Overbeck's approach to \ac{LRS} codes in the sum-rank metric.
	In contrast to the square-code distinguisher, the Overbeck-like distinguisher works for the general setting with an
	arbitrary automorphism $\aut$ and any valid $\aut$-derivation $\der$.
	Since it does not support block multipliers, i.e., \ac{GLRS} codes, let us quickly describe how we can apply
	distinguishers for \ac{LRS} codes to \ac{GLRS} codes in general.
	Recall therefore that a \ac{GLRS} code $\GLRS{\vecbeta, \a, \v; \n, k}$ has a generator matrix of the form
	$(v_1 \shot{\G}{1} \mid \dots \mid v_{\shots} \shot{\G}{\shots})$, where $\G \in \Fqm^{k \times n}$ is a generator
	matrix of $\linRS{\vecbeta, \a; \n, k}$.
	But this implies that the Overbeck-like distinguisher will (at least) succeed if we apply it to the matrix
	$(v_1^{-1} \shot{\M}{1} \mid \dots \mid v_{\shots}^{-1} \shot{\M}{\shots})$, where $\M \in \Fqm^{k \times n}$ denotes
	the public generator matrix of the \ac{GLRS} code.
	We can thus run the Overbeck-like distinguisher for different choices of $\v^{-1} \in \Fqm^{\shots}$ until it either
	succeeds or all possible $(q^m-1)^{\shots}$ (inverse) block multipliers were checked in the worst case.

	\begin{lemma}\label{lem:moore_matrix_properties}
	    Choose $k < n$, let the entries of $\a \in \Fqm^{\shots}$ belong to distinct nontrivial conjugacy classes of
		$\Fqm$ and let $\x \in \Fqm^{n}$ be a vector with $\SumRankWeight(\x) = n$.
		Then the following holds for the generalized Moore matrix $\opMoore{k}{\x}{\a}$:
		\begin{enumerate}
			\item The addition code $\mycode{A} \defeq \langle \opMoore{k}{\x}{\a} \rangle + \langle \op{\a}{\opMoore{k}{\x}{\a}} \rangle$
					equals $\langle \opMoore{k+1}{\x}{\a} \rangle$ and thus $\dim (\mycode{A}) = k + 1$.
		    \item The intersection code $\mycode{I} \defeq \langle \opMoore{k}{\x}{\a} \rangle \cap \langle \op{\a}{\opMoore{k}{\x}{\a}} \rangle$
					is generated by the matrix $\opMoore{k-1}{\op{\a}{\x}}{\a}$ and hence $\dim (\mycode{I}) = k - 1$.
		\end{enumerate}
	\end{lemma}

	\begin{proof}
		\begin{enumerate}
		    \item Let $\A \in \Fqm^{2k \times n}$ denote the matrix that is obtained by vertically stacking
					$\opMoore{k}{\x}{\a}$ and $\op{\a}{\opMoore{k}{\x}{\a}}$.
					Since the first $k-1$ lines of $\op{\a}{\opMoore{k}{\x}{\a}}$ coincide with the last $k-1$ rows of
					$\opMoore{k}{\x}{\a}$ due to the Moore-matrix structure, we obtain
					\begin{equation}
						\mycode{A} = \langle \A \rangle = \left\langle
						\begin{pmatrix}
							\opMoore{k}{\x}{\a} \\
							\opexp{\a}{\x}{k}
						\end{pmatrix}
						\right\rangle
						= \langle \opMoore{k+1}{\x}{\a} \rangle.
					\end{equation}
					As further $k+1 \leq n$ holds and the necessary conditions on $\a$ and $\x$ apply, we get
					$\dim (\mycode{A}) = \rk_{q^m}(\opMoore{k+1}{\x}{\a}) = \min (k + 1, n) = k + 1$.
			\item As the last $k-1$ lines of $\opMoore{k}{\x}{\a}$ and the first $k-1$ lines of
					$\op{\a}{\opMoore{k}{\x}{\a}}$ coincide, their span $\langle \opMoore{k-1}{\op{\a}{\x}}{\a} \rangle$
					is certainly contained in $\mycode{I}$.
					Note that, because of the $\Fq$-linearity of $\op{a}{\cdot}$,
					$\SumRankWeight(\op{\a}{\x}) = \SumRankWeight(\x) = n$ holds,
					which implies $\rk_{q^m}(\opMoore{k-1}{\op{\a}{\x}}{\a}) = k - 1$.
					Thus,
					\begin{align}
						\dim(\mycode{I}) &= \rk_{q^m}(\opMoore{k}{\x}{\a}) + \rk_{q^m}(\op{\a}{\opMoore{k}{\x}{\a}})
						- \dim(\mycode{A}) \\
						&= 2k - k - 1 = k - 1
					\end{align}
					and $\mycode{I} = \langle \opMoore{k-1}{\op{\a}{\x}}{\a} \rangle$ follows from the dimension equality.
		\end{enumerate}
	\qed
	\end{proof}

	Define the operator
	\begin{equation}\label{eq:def_full_operator}
		\Gamma_{\a}^j: \quad \Fqm^{k \times n} \to \Fqm^{(j+1)k \times n},
		\qquad \M \mapsto
		\begin{pmatrix}
			\M
			\\
			\op{\a}{\M}
			\\
			\vdots
			\\
			\opexp{\a}{\M}{j}
		\end{pmatrix}
	\end{equation}
	for a fixed vector $\a \in \Fqm^{\shots}$ of evaluation parameters and a natural number $j \in \NN$.

	\begin{corollary}\label{cor:overbeck_distinguisher}
		Let $\G$ be an arbitrary generator matrix of the code $\linRS{\vecbeta, \a; \n, k}$.
		Then, $\Gamma_{\a}^{j}(\G)$ generates the code $\linRS{\vecbeta, \a; \n, k + j}$ and $\rk_{q^m}(\Gamma_{\a}^{j}(\G)) = k + j$ holds
		for all $0 \leq j \leq n - k$.
	\end{corollary}

	\begin{proof}
		If $\G = \opMoore{k}{\vecbeta}{\a}$, the statements follow from  an iterative application of \autoref{lem:moore_matrix_properties}.
	    In any other case, there is a matrix $\S = (S_{i, j})_{i, j} \in \GL_{k}(\Fqm)$ such that $\G = \S \cdot \opMoore{k}{\vecbeta}{\a}$.

		Let us first focus on the smallest nontrivial choice for $j$, namely $j = 1$.
		The $l$-th row of $\op{\a}{\G} = \op{\a}{\S \cdot \opMoore{k}{\vecbeta}{\a}}$ is
		\begin{align}
		    \opLargeParens{\a}{\sum_{i=1}^{k} S_{l,i} \opexp{\a}{\vecbeta}{i-1}}
			&= \sum_{i=1}^{k} \op{\a}{S_{l,i} \opexp{\a}{\vecbeta}{i-1}} \\
			&\overset{(\ast)}{=} \sum_{i=1}^{k} \aut(S_{l,i}) \opexp{\a}{\vecbeta}{i} + \der(S_{l,i}) \opexp{\a}{\vecbeta}{i-1},
			\label{eq:overbeck_row_representation}
		\end{align}
		where $(\ast)$ follows from~\autoref{lem:opProductRule}.
		But this is a $\Fqm$-linear combination of the elements $\vecbeta, \op{\a}{\vecbeta}, \dots, \opexp{\a}{\vecbeta}{k}$,
		i.e., of a basis of $\linRS{\vecbeta, \a; \n, k + 1}$.
		Hence, the inclusion $\langle \op{\a}{\G} \rangle \subseteq \linRS{\vecbeta, \a; \n, k + 1}$ applies.
		Since $\G$ generates $\linRS{\vecbeta, \a; \n, k} \subseteq \linRS{\vecbeta, \a; \n, k + 1}$, it follows further
		that $\langle \Gamma_{\a}(\G) \rangle \subseteq \linRS{\vecbeta, \a; \n, k + 1}$.

		Let us show the other inclusion $\langle \Gamma_{\a}(\G) \rangle \supseteq \linRS{\vecbeta, \a; \n, k + 1}$.
		First realize that $\langle \Gamma_{\a}(\G) \rangle \supseteq \langle \G \rangle = \linRS{\vecbeta, \a; \n, k}$
		and $\linRS{\vecbeta, \a; \n, k + 1} = \linRS{\vecbeta, \a; \n, k} + \langle \opexp{\a}{\vecbeta}{k} \rangle$ hold.
		It is thus enough to show that there is an element of $\langle \Gamma_{\a}(\G) \rangle$ whose $\Fqm$-linear
		combination contains a nonzero multiple of $\opexp{\a}{\vecbeta}{k}$.
		But since $\S$ has full rank, there is a nonzero entry in its $k$-th column, say $S_{l^{\ast}, k}$.
		Now~\eqref{eq:overbeck_row_representation} shows that the $l^{\ast}$-th row of $\op{\a}{\G}$ has the form
		\begin{equation}
		    \aut(S_{l^{\ast}, k}) \opexp{\a}{\vecbeta}{k} + \sum_{i=1}^{k-1} \aut(S_{l,i}) \opexp{\a}{\vecbeta}{i} + \der(S_{l,i}) \opexp{\a}{\vecbeta}{i-1},
		\end{equation}
		where the right-hand side is clearly contained in $\linRS{\vecbeta, \a; \n, k}$.
		As $\aut(S_{l^{\ast}, k})$ is nonzero if and only if $S_{l^{\ast}, k} \neq 0$, this shows
		$\langle \op{\a}{\G} \rangle \supseteq \langle \opexp{\a}{\vecbeta}{k} \rangle$ and hence
		$\langle \Gamma_{\a}(\G) \rangle \supseteq \linRS{\vecbeta, \a; \n, k + 1}$.

		Summing up, we obtain $\langle \Gamma_{\a}(\G) \rangle = \linRS{\vecbeta, \a; \n, k + 1}$, which directly implies
		$\rk_{q^m}(\Gamma_{\a}(\G)) = k + 1$.

		For $j > 1$, the results follow inductively from the fact that
		\begin{equation}
		    \langle \Gamma_{\a}^{j}(\G) \rangle = \left\langle
			\begin{pmatrix}
				\Gamma_{\a}^{j-1}(\G) \\
				\opexp{\a}{\G}{j}
			\end{pmatrix}
			\right\rangle \overset{(\circ)}{=} \langle \Gamma_{\a}(\Gamma_{\a}^{j-1}(\G)) \rangle,
		\end{equation}
		since all rows that are added in step $(\circ)$ are already contained in the row space of $\Gamma_{\a}^{j-1}(\G)$.
		The statements $\langle \Gamma_{\a}^{j}(\G) \rangle = \linRS{\vecbeta, \a; \n, k + j}$ and hence
		$\rk_{q^m}(\Gamma_{\a}^{j}(\G)) = k + j$ follow with the knowledge of
		$\Gamma_{\a}^{j-1}(\G) = \linRS{\vecbeta, \a; \n, k + j - 1}$ and the proof for $j = 1$.
	\qed
	\end{proof}

	In contrast, randomly chosen full-rank matrices over $\Fqm$ tend to behave quite differently when $\Gamma_{\a}$ is
	applied.
	This is analogous to~\cite[Assumption~2]{Ove2005}.

	\begin{conjecture}\label{conj:rank_under_full_operator}
		Let $\M \in \Fqm^{k \times n}$ be a randomly chosen matrix with full $\Fqm$-rank and such that each block
		$\shot{\M}{i}$ for $i = 1, \dots, \shots$ has full column rank over $\Fq$.
		Assume that $\a \in \Fqm^{\shots}$ consists of randomly chosen representatives of distinct nontrivial conjugacy
		classes of $\Fqm$ and fix a parameter $j \in \{1, \dots, \len - k\}$.
		Then, $\rk_{q^m}(\Gamma_{\a}^j(\M)) = \min((j+1)k, \len)$ holds with high probability.
	\end{conjecture}

	With these results, we can solve~\autoref{prob:distinguish} for \ac{LRS} codes in polynomial time if $\aut$, $\der$,
	and $\a$ are known.
	We summarize it in~\autoref{thm:overbeck_dist}:

	\begin{theorem}[Overbeck-like Distinguisher]\label{thm:overbeck_dist}
		Let $\M \in \Fqm^{k \times n}$ be an arbitrary full-rank matrix.
		We can decide with high probability\footnote{In fact, the distinguisher recognizes a
		\ac{GLRS} code with probability one. But, with a small probability, it might wrongly declare a non-\ac{GLRS} code to be a \ac{GLRS} code.}
		if $\M$ generates an \ac{LRS} code with respect to $\aut$, $\der$, and $\a$
		in $\OCompl{n^5}$ operations in $\Fqm$.
	\end{theorem}

	\begin{proof}
		First, choose a $0 \leq j \leq n - k$ for which $k + j < \min((j+1)k, n)$ holds.
		We set up the matrix $\Gamma_{\a}^{j}(\M) \in \Fqm^{(j+1)k \times n}$ in $\OCompl{jkn} \subseteq \OCompl{n^3}$
		operations in $\Fqm$.
		Next, we compute its rank in $\OCompl{n^5}$ $\Fqm$-operations.
		By~\autoref{cor:overbeck_distinguisher} and~\autoref{conj:rank_under_full_operator}, we know with high
		probability that $\M$ generates an \ac{LRS} code with respect to the given parameters if
		$\rk_{q^m}(\Gamma_{\a}^{j}(\M)) = k + j$ holds.
		If however $\rk_{q^m}(\Gamma_{\a}^{j}(\M)) > k + j$, we know for sure that $\M$ is no generator matrix of an
		\ac{LRS} code with respect to the given parameters.
	\qed
	\end{proof}

	\begin{remark}
	    We empirically verified by simulations that the distinguisher can in most cases \emph{not} recognize an \ac{LRS} code if it is executed with respect to a
		different set of evaluation parameters.
		This is the case even if the conjugacy classes of the evaluation parameters $\a = (a_1, \dots, a_{\shots})$
		are known and only other representatives $\hat{\a} \defeq (\conj{a_1}{c_1}, \dots, \conj{a_{\shots}}{c_{\shots}})$
		with $c_1, \dots, c_{\shots} \in \Fqm^{\ast}$ are used for the distinguisher.

		This means that not even side information about the chosen conjugacy classes helps the distinguishing process
		but knowledge of the exact evaluation parameters is needed.
		If we do not have access to this information, we have to try exponentially many possibilities in the worst case.
	\end{remark}

	We use the remainder of this section to give a short outline of how the distinguisher using the \acs{HMR} approach
	can be generalized to the \ac{LRS} case.
	If a full-rank matrix $\M \in \Fqm^{k \times n}$ and parameters $\aut$, $\der$, and $\a$ are given, we focus on the
	intersection code instead of considering the addition code as we indicated earlier.
	Similar to~\autoref{cor:overbeck_distinguisher}, we can derive~\autoref{cor:overbeck_intersection_distinguisher}
	whose proof we omit for brevity.

	\begin{corollary}\label{cor:overbeck_intersection_distinguisher}
		Let $\G$ be an arbitrary generator matrix of the code $\linRS{\vecbeta, \a; \n, k}$.
		Then, the $j$-fold intersection code $\bigcap_{i=0}^{j} \langle \opexp{\a}{\M}{i} \rangle$ equals
		$\linRS{\opexp{\a}{\vecbeta}{j}, \a; \n, k - j}$ and has thus $\Fqm$-dimension $k - j$ for all $0 \leq j \leq k - 1$.
	\end{corollary}

	Heuristically speaking, the application of $\op{\a}{\cdot}$ to a random full-rank matrix produces another
	essentially random code.
	For small dimension $k$, it is hence reasonable to assume that the $j$-fold intersection code
	from~\autoref{cor:overbeck_intersection_distinguisher} has a much lower dimension.
	This illustrates why~\autoref{cor:overbeck_intersection_distinguisher} can serve as a
	distinguisher for \ac{LRS} codes that can, of course, also be applied to \ac{GLRS} codes as explained in the
	beginning of this section.

	\section{Recovery of a Canonical Generator Matrix}\label{sec:decoding}

	If only a scrambled and possibly further disguised generator matrix of a \ac{GLRS} code is known, it is a crucial task
	to recover a canonical generator matrix of the same code.
	The secret code structure, that is revealed by a canonical generator matrix, is (up to now) directly linked to the
	knowledge of efficient decoding algorithms.
	We partly tackle~\autoref{prob:recover} in this section and show how the recovery can be done in the case of \ac{GLRS} codes with zero derivation for which
	the automorphism $\aut$ is given.
	As for the distinguishers, the following results were also implemented in SageMath and checked for several parameter sets.

	The first approach requires the identity automorphism and finds suitable evaluation parameters $\a$ and block multipliers $\v$, whereas
	the second one assumes the knowledge of $\a$ and allows to recover $\vecbeta$ for an arbitrary but known automorphism.
	If \ac{GLRS} codes with respect to the identity automorphism are considered, we can thus combine the two distinguishers to recover first $\a$ and $\v$, and then
	$\vecbeta$.

	\subsection{Square-Code Approach}
	
	For this approach, we focus on the identity automorphism which allows zero derivation only.
	The recovery strategy is based on the fact that we can extract a \ac{GRS} code from an arbitrary generator matrix of
	a \ac{GLRS} code as described in~\autoref{sec:square_distinguisher}.
	We then recover the parameters of the \ac{GRS} code and afterwards the ones of the \ac{GLRS} code.

	\begin{theorem}\label{thm:grs_recovery}
		Let $\G \in \Fqm^{k \times n}$ denote a generator matrix of a \ac{GLRS} code $\mycode{C}$ with respect to the
		identity automorphism and zero derivation.
		We can recover parameters $\a, \v \in \Fqm^{\shots}$ for which a canonical generator matrix of $\mycode{C}$ exists in $\OCompl{k^2 n}$ operations in $\Fqm$.
	\end{theorem}

	\begin{proof}
		Recall from~\autoref{lem:square_code_dim} that the matrix consisting of one column of each block $\shot{\G}{i}$ of the generator matrix $\G$ generates a GRS code of length $\ell$ and dimension $k$. If $\ell > k$ this is a nontrivial GRS code, whereas for $\ell \leq k$ the code is the whole space $\mathbb F_{q^m}^\ell$. For the description of the recovery process we will differentiate between these two cases:
		\begin{enumerate}
		    \item In the case where $\ell \geq k$ holds we can simply choose the evaluation points to be $\a = (1,\alpha, \alpha^2,\dots, \alpha^\ell)$ for a primitive element $\alpha \in \mathbb F_{q^m}$. Clearly, the Vandermonde matrix with these parameters is full-rank and hence spans the whole space, i.e., it is a generator matrix of the trivial RS code. Furthermore, $1,\alpha, \alpha^2,\dots, \alpha^\ell$ represent distinct conjugacy classes (since we consider zero derivation) and are hence a valid choice for the LRS code. Note that we do not need to consider column multipliers in this setting, i.e., we can assume $\v$ to be the all-one vector.
			\item In the other case, i.e., where $\ell <k$, the resulting GRS code and its dual code are nontrivial, i.e., they both have minimum distance greater than one. We can now use the Sidelnikov--Shestakov algorithm from~\cite{sidelnikov1992insecurity} on the parity-check matrix of our GRS code to find suitable $\a$ and $\v \in \Fqm^\shots$.
			This requires $\OCompl{k^2 n}$ operations in $\Fqm$.
		\end{enumerate}
	\qed
	\end{proof}
Depending on how the code is disguised in a potential cryptosystem, an attacker can use the fact about the square-code dimension from~\autoref{lem:square_code_dim} to find suitable subcodes of the public code.
	Then, the parameter-recovery algorithm from~\autoref{thm:grs_recovery} can be applied to the obtained subcodes.

	\subsection{Overbeck-like Approach}

	In the literature, there are three different approaches for recovering the secret parameters of Gabidulin codes based
	on ideas similar to Overbeck's distinguisher:
	\begin{enumerate}
	    \item Overbeck~\cite{Ove2005} considers the sum of the codes obtained by repeated application of the Frobenius
				automorphism until a code of codimension one is obtained.
				The secret parameters are then recovered from a generator of the one-dimensional dual code.
		\item \acl{HaMaR}~\cite{horlemann2018extension} compute the intersection of the codes that arise from repeated
				application of the Frobenius automorphism until the result is a one-dimensional code.
				A generator of the latter yields the secret parameters of the code.
		\item Another approach by \acl{HaMaR}~\cite{horlemann2016rankBased} maps the task to the problem of finding
				rank-one codewords in the code generated by the public matrix and a corrupted codeword.
	\end{enumerate}

	We present the first two approaches for \ac{LRS} codes in the zero-derivation regime where the automorphism $\aut$
	and the evaluation parameters $\a \in \Fqm^{\shots}$ are known.
	Note that the third technique is also applicable to our setting but omitted for brevity.
	Moreover, the recovery methods extend to \ac{GLRS} codes by executing them after guessing the block multipliers,
	similar to the distinguishing strategy explained in~\autoref{sec:overbeck_distinguisher}.

	\begin{theorem}
		Let $\G \in \Fqm^{k \times n}$ denote a generator matrix of an \ac{LRS} code $\mycode{C} \defeq \linRS{\vecbeta, \a; \n, k}$ with respect to a known
		automorphism $\aut$ and zero derivation.
		If the evaluation parameters $\a \in \Fqm^{\shots}$ are known, we can recover code locators $\tilde{\vecbeta} \in \Fqm^{n}$
		such that $\opMoore{k}{\tilde{\vecbeta}}{\a}$ generates $\mycode{C}$ in $\OCompl{n^5}$ operations in $\Fqm$.
	\end{theorem}

	\begin{proof}
		First note that any $\Fqm^{\ast}$-multiple $\tilde{\vecbeta}$ of $\vecbeta$ is sufficient because
		\begin{equation}
			\opMoore{k}{\tilde{\vecbeta}}{a} = \diag\left((c, \aut(c), \dots, \aut^{k-1}(c))\right) \cdot \opMoore{k}{\vecbeta}{\a}
		\end{equation}
		holds for $\tilde{\vecbeta} \defeq c \cdot \vecbeta$ with $c \in \Fqm^{\ast}$.
		Since the diagonal matrix has full rank, the row spaces of $\opMoore{k}{\tilde{\vecbeta}}{\a}$ and
		$\opMoore{k}{\vecbeta}{\a}$ both equal $\mycode{C}$.
		We show how to recover such a $\tilde{\vecbeta} \in \Fqm^{n}$ with the first two of the three approaches mentioned above:
		\begin{enumerate}
		    \item
			From~\autoref{cor:overbeck_distinguisher}, we obtain the equality
			$\langle \Gamma_{\a}^{n-k-1}(\G) \rangle = \linRS{\vecbeta, \a; \n, n - 1}$ and the dual $\mycode{D}$ of
			this code has dimension one.
			The solution $\H \in \Fqm^{n}$ of the system $\Gamma_{\a}^{n-k-1}(\G) \cdot \H^{\top} = \0$ is a generator
			matrix (or rather a generator vector) of $\mycode{D}$.
			Since we are in the zero-derivation case, we can use the result~\eqref{eq:lrs_dual_zero_der} about duals of
			\ac{LRS} codes and recover a suitable $\tilde{\vecbeta}$ from $\H$ via~\eqref{eq:syndrome_eqs_lrs}.
			\item
			We first compute the intersection space $\bigcap_{i=0}^{k-1} \langle \opexp{\a}{\G}{i} \rangle$ which is equal to
			$\linRS{\opexp{\a}{\vecbeta}{k-1}, \a; \n, 1}$ according to~\autoref{cor:overbeck_intersection_distinguisher}.
			Therefore, every generator $\g \in \Fqm^{n}$ of this space (and in particular the one that we computed) has
			the form $c \cdot \opexp{\a}{\vecbeta}{k-1}$ for a $c \in \Fqm^{\ast}$.
			Note that, in the zero-derivation case, the inverse of the operator $\opexp{a}{\cdot}{i}$ for fixed $a \in \Fqm$ and $i \geq 0$ is
			\begin{equation}
			    \opexpinv{a}{b}{i} \defeq \aut^{-i}\left(\frac{b}{\genNorm{i}{a}}\right)
				\quad \text{for all } b \in \Fqm.
			\end{equation}
			We use this fact to derive the following equation from $\g = c \cdot \opexp{\a}{\vecbeta}{k-1}$:
			\begin{equation}
			    \aut^{-k+1}\left( \left( \frac{\shot{\g}{1}}{\genNorm{k-1}{a_1}} \mid \dots \mid \frac{\shot{\g}{\shots}}{\genNorm{k-1}{a_{\shots}}} \right) \right)
				= \aut^{-k+1}(c) \cdot \vecbeta.
			\end{equation}
			Solving the obtained system of linear equations lets us recover a suitable $\tilde{\vecbeta}$.
		\end{enumerate}
		The complexity is in both cases dominated by computing the reduced row-echelon form of $\Gamma_{\a}^{n-k-1}(\G)$
		and $\Gamma_{\a}^{k-1}(\G)$, respectively.
		This can be achieved in $\OCompl{n^5}$ operations in $\Fqm$.
	\qed
	\end{proof}

	\section{Conclusion}

	We introduced \ac{GLRS} codes as \ac{LRS} codes with nonzero block multipliers and proposed two distinguishers for this code family that are inspired by similar techniques in the Hamming and the rank
	metric.
	The square-code distinguisher works for the identity automorphism and zero derivation, whereas the
	Overbeck-like distinguisher can handle arbitrary automorphisms and derivations.
	Both have polynomial runtime when the automorphism $\aut$, the derivation $\der$, and in the latter case additionally the evaluation parameters $\a$ and the
	block multipliers $\v$ are known.

	Since many McEliece-like cryptosystems use isometric disguising, we further studied codes that are semilinearly
	equivalent to \ac{GLRS} codes.
	We showed that \ac{GLRS} codes are closed under semilinear
	equivalence for a fixed automorphism and some possible choices for the derivation.

	Finally, we partially solved the problem of recovering a canonical generator matrix (and thus finding an efficient decoder) from an
	arbitrary generator matrix of a \ac{GLRS} code in the zero-derivation case.
	The complexity is again polynomial if either $\aut=\mathrm{Id}$ or $\aut$, $\v$ and $\a$ are known.
	More precisely, we showed that the square-code code approach allows to recover suitable evaluation parameters $\a$ and block multipliers $\v$ of a GLRS code
	in the identity-automorphism setting, and that an Overbeck-like strategy can recover suitable code locators $\vecbeta$ of
	a GLRS code for arbitrary automorphisms and zero derivations if $\a$ and $\v$ are known.

	This work is a first step towards building quantum-secure cryptosystems in the sum-rank metric.
	Naturally, many other research questions arise in this field:
	As simulations show, the Overbeck-like distinguisher seems not to work when
	the wrong evaluation parameters are used.
	This is the case even when the parameters are chosen from the correct conjugacy classes, what makes it interesting to
	study.
	Another idea is to find a new operation with respect to which the square-code distinguisher works also for
	arbitrary automorphisms.

	We further want to investigate more distinguishing methods as e.g.\ augmenting the generator matrix or applying
	near-isometries and see also how \ac{GLRS} codes and their distinguishers carry over to the skew metric.

	\vspace{-.1cm}

    \bibliographystyle{splncs04}
% Added from bbl:

\end{document}